\documentclass[conference]{IEEEtran}
\IEEEoverridecommandlockouts
\usepackage{cite}
\usepackage{amsmath,amssymb,amsfonts,amsthm}
\usepackage{algorithmic}
\usepackage{graphicx}
\usepackage{textcomp}
\usepackage{xcolor}

\newtheorem{thm}{Theorem}
\newtheorem{lem}{Lemma}

\newtheorem{cor}{Corollary}
\newtheorem{remark}{Remark}
\theoremstyle{definition}
\newtheorem{definition}{Definition}

\def\BibTeX{{\rm B\kern-.05em{\sc i\kern-.025em b}\kern-.08em
    T\kern-.1667em\lower.7ex\hbox{E}\kern-.125emX}}
\begin{document}

\title{Sub-Gaussian Error Bounds for Hypothesis Testing\\
\thanks{This research was supported in part by the US National Science Foundation under grant HDR: TRIPODS 19-34884.}
}

\author{\IEEEauthorblockN{Yan Wang}
\IEEEauthorblockA{
Department of Statistics, Iowa State University\\
Ames, IA 50011, USA \\
Email: wangyan@iastate.edu}
}

\maketitle

\begin{abstract}
We interpret likelihood-based test functions from a geometric perspective where the Kullback-Leibler (KL) divergence is adopted to quantify the distance from a distribution to another. Such a test function can be seen as a sub-Gaussian random variable, and we propose a principled way to calculate its corresponding sub-Gaussian norm. Then an error bound for binary hypothesis testing can be obtained in terms of the sub-Gaussian norm and the KL divergence, which is more informative than Pinsker's bound when the significance level is prescribed. For $M$-ary hypothesis testing, we also derive an error bound which is complementary to Fano's inequality by being more informative when the number of hypotheses or the sample size is not large.
\end{abstract}


\section{Introduction}
Hypothesis testing is one central task in statistics. One of its simplest forms is the binary case: given $n$ independent and identically distributed (i.i.d.) random variables $X_1^n\equiv(X_1,\ldots,X_n)$, one wants to infer whether the null hypothesis $H_0:X_i\sim P_0$ or the alternative hypothesis $H_1:X_i\sim P_1$ is true. The binary case serves as an important starting point from which further results can be established, in the settings of both classical and quantum hypothesis testing \cite{Book:Keener,Book:Hayashi}. With $X_1^n$, one can construct the empirical distribution $\hat P_n=\frac{1}{n}\sum_{i=1}^n\delta_{X_i}$, where $\delta_X$ is the Dirac measure that puts unit mass at $X$. Adopting the Kullback-Leibler (KL) divergence as a distance from $\hat P_n$ to $P_0$ or $P_1$, one can construct a test function as
\begin{align}\label{eq:Phi}
\Phi(X_1^n) = I\{D_\mathrm{KL}(\hat P_n\rVert P_0) - D_\mathrm{KL}(\hat P_n\rVert P_1) > c\},
\end{align}
where $I\{\cdot\}$ is the indicator function, $c\geq0$ serves as a threshold beyond which the decision that $\hat P_n$ is closer to $P_1$ than to $P_0$ is made, and $D_\mathrm{KL}(P\rVert Q)=\int\ln(dP/dQ)dP$ is the KL divergence from probability $P$ to probability $Q$ if $P\ll Q$. Conventionally, if $P$ is not absolutely continuous with respect to $Q$, then $D_\mathrm{KL}(P\rVert Q)\equiv\infty$. Note $\hat P_n$ is discrete; hence if both $P_0$ and $P_1$ are discrete with the same support, \eqref{eq:Phi} is well defined. Denote the densities of $P_0$ and $P_1$ with respect to the counting measure as $p_0$ and $p_1$, respectively, and we have
\begin{align}\label{eq:KL_diff}
D_\mathrm{KL}(\hat P_n\rVert P_0) - D_\mathrm{KL}(\hat P_n\rVert P_1) = \frac{1}{n}\ln\left(\frac{\prod_{i=1}^n p_1(X_i)}{\prod_{i=1}^n p_0(X_i)}\right).
\end{align}
In fact, in this case, \eqref{eq:Phi} is equivalent to the test function for the likelihood ratio test \cite{Book:Cover}
\begin{align}\label{eq:Phi_lrt}
\Phi_\mathrm{lrt}(X_1^n) = I\left\{ \frac{\prod_{i=1}^n p_1(X_i)}{\prod_{i=1}^n p_0(X_i)} > c'\right\},
\end{align}
where $c'=e^{cn}$. In the case that both $P_0$ and $P_1$ are continuous, the KL divergence difference $D_\mathrm{KL}(\hat P_n\rVert P_0)-D_\mathrm{KL}(\hat P_n\rVert P_1)$ is not well defined. Nonetheless, the technically tricky part is the term ``$\int \hat p_n\ln(\hat p_n) d\mu$,'' where we use $\hat p_n$ to denote the density of $\hat P_n$ with respect to the Lebesgue measure $\mu$ as if it had one. But it appears twice and is cancelled out formally. We might conveniently \emph{define} the KL divergence difference in this case as \eqref{eq:KL_diff}, and still find the equivalence between \eqref{eq:Phi} and \eqref{eq:Phi_lrt}. Using the KL divergence in the context of hypothesis testing can be beneficial. Firstly, it provides a clear geometric meaning to the likelihood ratio test, as well as to the general idea underlying hypothesis testing. Secondly, it also offers a geometric, or even physical, interpretation of the lower bound for the resulting statistical errors, as shown below.

Under the null hypothesis $H_0$, the type I error rate (or the significance level) $\alpha$ that is incurred by applying \eqref{eq:Phi} for a fixed $c$ is
\begin{align}\label{eq:alpha}
\alpha = \mathbb{E}_{X_1^n\sim P_0^{\otimes n}} \Phi(X_1^n),
\end{align}
where $P_0^{\otimes n}$ is the product probability measure for $X_1^n$ under $H_0$. In practice, by prescribing the significance level, for example, letting $\alpha=0.05$, one can derive the corresponding $c$ and determine the desired test function. However, in this work, our focus is $not$ to find a test function at given $\alpha$, we mainly deal with the case that $c$ is fixed, and $\alpha$ is obtained in a somewhat passive way. Thanks to the Neyman-Pearson lemma \cite{paper:NP-Lemma}, the likelihood ratio test is known to be optimal in the sense of statistical power. Hence, given the incurred $\alpha$, test function \eqref{eq:Phi} has the minimal type II error rate $\beta$ among all possible test functions with the corresponding type I error rate no gretaer than $\alpha$:
\begin{align}\label{eq:beta}
\beta = 1 - \mathbb{E}_{X_1^n\sim P_1^{\otimes n}} \Phi(X_1^n),
\end{align}
where $P_1^{\otimes n}$ is the product probability measure for $X_1^n$ under the alternative hypothesis $H_1$.

Controlling statistical errors is of practical importance; however, typically one cannot suppress both types of error simultaneously. Under our i.i.d. setting, a classical result, based on Pinsker's inequality, concerning the error bound for any (measurable) test function is that \cite{Book:Cover}
\begin{align}\label{eq:Pinsker_result}
\alpha + \beta \geq 1-\sqrt{\frac{n}{2}D_\mathrm{KL}(P_1\rVert P_0)}.
\end{align}
This result is striking in that without going into the details of calculating $\alpha$ and $\beta$, one can have a \emph{nontrivial} lower bound of their sum in terms of the KL divergence between two candidate probabilities, as long as the right-hand side of \eqref{eq:Pinsker_result} is greater than 0. For a fixed $n$, this bound is solely determined by $D_\mathrm{KL}(P_1\rVert P_0)$, which reflects the ``distance'' from $P_1$ to $P_0$. This result also has a significant physical meaning. At a nonequilibrium steady state, if $P_1$ denotes the probability associated with observing a stochastic trajectory in the forward process, and $P_0$ in the backward process, then the theory of stochastic thermodynamics tells us that $D_\mathrm{KL}(P_1\rVert P_0)$ is equivalent to the average entropy production $\Delta S$ in the forward process, which is always nonnegative \cite{PRL:Seifert2005,me}. Hence, if one wishes to infer the arrow of time based on observations, then Pinsker's result \eqref{eq:Pinsker_result} implies that the chance of making an error is high if $\Delta S$ is small. In fact, we know that $\Delta S=0$ at equilibrium, and one cannot tell the arrow of time at all; hypothesis testing is just random guess in this case. 

While \eqref{eq:Pinsker_result} is useful, can we have a tighter and thus more informative bound? In this work, we will show that by taking advantage of the sub-Gaussian property of $\Phi(X_1^n)$ \cite{Book:Martin,Book:Roman}, one can derive a bound \eqref{eq:error_sum} on statistical errors in terms of its sub-Gaussian norm (as well as the KL divergence from $P_1$ to $P_0$). We name such an error bound as ``sub-Gaussian'' to highlight this fact. It turns out that it is a tighter bound than \eqref{eq:Pinsker_result} in the sense that it provides a greater lower bound for $\alpha+\beta$ (or for $\beta$ at any given $\alpha\neq 0.5$). In practice, a small $\alpha$ is commonly set as the significance level, and our result can hopefully be more relevant. Moreover, in the case of $M$-ary hypothesis testing where $M>2$ hypotheses are present, we also derive a bound \eqref{eq:beyond_fano} for making incorrect decisions, which is complementary to the celebrated Fano's inequality \cite{Book:Fano} when the number of hypotheses $M$ or the sample size $n$ is not large. The error bounds presented in this work are universal and easily applicable. We hope these findings can help better quantify errors in various statistical practices involving hypothesis testing.

\section{Main Results}
We will first introduce the sub-Gaussian norm of $\Phi(X_1^n)$. Then error bounds in the binary and $M$-ary cases are established, respectively.

\subsection{Sub-Gaussian norm of $\Phi(X_1^n)$}

Sub-Gaussian random variables are natural generalizations of Gaussian ones. The so-called sub-Gaussian property can be defined in several different but equivalent ways \cite{Book:Martin,Book:Roman}. In this work, we pick one that suits most for our purposes.

\begin{definition}
A random variable $X$ with probability law $P$ is called sub-Gaussian if there exists $\sigma>0$ such that its central moment generating function satisfies
\begin{align}
\mathbb{E}_Pe^{s(X-\mathbb{E}_PX)}\leq e^{\sigma^2s^2/2},\ \forall s \in \mathbb{R}.\notag
\end{align}
\end{definition}

\begin{definition}
The associated sub-Gaussian norm $\sigma_{XP}$ of $X$ with respect to $P$ is defined as
\begin{align}
\sigma_{XP}\equiv\inf\{\sigma>0:\mathbb{E}_Pe^{s(X-\mathbb{E}_PX)}\leq e^{\sigma^2s^2/2},\ \forall s \in \mathbb{R}\}.\notag
\end{align}
\end{definition}

\begin{remark}
$\sigma_{XP}$ is a well defined norm for the centered variable $X-\mathbb{E}_PX$ \cite{me}. It is the same for a location family of random variables that have different means but are otherwise identical. Also, $\sigma_{XP}$ is equal to the $\psi_2$-Orlicz norm of $X-\mathbb{E}_PX$ up to a numerical constant factor.
\end{remark}

\begin{lem}\label{lem:boundedsubG}
A bounded random variable is sub-Gaussian. In particular, if $X\in[a,b]$ almost surely with respect to $P$, then $\sigma_{XP}\leq (b-a)/2$. 
\end{lem}
\begin{proof}
This is a well known result that can be found in, for example, \cite{Book:Martin,Book:Roman}.
\end{proof}

Test function \eqref{eq:Phi} is an indicator function and takes on values in $\{0,1\}$; hence it is bounded. No matter what the law of $X_1^n$ is, $\Phi(X_1^n)$ is always sub-Gaussian by Lemma \ref{lem:boundedsubG} with a uniform upper bound of its sub-Gaussian norm that
\begin{align}\label{eq:Phi_norm_0.5}
\sigma_{\Phi P}\leq 0.5.
\end{align}
However, if $\alpha$ is fixed as a result of some $c$ being used in \eqref{eq:Phi}, then a more informative sub-Gaussian norm for $\Phi(X_1^n)$ can be obtained under the situation that $X_1^n\sim P_0^{\otimes n}$. In this case, by \eqref{eq:alpha},
\begin{align}
\text{Pr}(\Phi(X_1^n)=1|H_0) = \mathbb{E}_{X_1^n\sim P_0^{\otimes n}} \Phi(X_1^n) = \alpha,\notag
\end{align}
and one can explicitly write
\begin{align}
&\mathbb{E}_{X_1^n\sim P_0^{\otimes n}} e^{s[\Phi(X_1^n) - \alpha]}\notag\\
=\ & \text{Pr}(\Phi=1|H_0) e^{s(1-\alpha)} + \text{Pr}(\Phi=0|H_0) e^{s(0-\alpha)} \notag\\
=\ &\alpha e^{s(1-\alpha)} + (1-\alpha)e^{-s\alpha} \equiv e^{f}.\notag
\end{align}
Using $f$, one can rewrite the sub-Gaussian property as
\begin{align}\label{eq:f}
h \equiv f-\frac{1}{2}\sigma^2s^2\leq 0,\ \forall s\in\mathbb{R}.
\end{align}
Since $\Phi$ is sub-Gaussian, there exists $\sigma$ such that at any $\alpha$, we have $h(s=0)=0$, which is the maximal value of $h$. This fact implies $\partial h/\partial s|_{s=0}=0$ and $\partial^2 h/\partial s^2|_{s=0}\leq0$. The latter poses a constraint on $\sigma$'s under which \eqref{eq:f} holds:
\begin{align}\label{eq:sigma_range}
\partial^2 h/\partial s^2|_{s=0}\leq0\Longrightarrow \sigma^2 \geq \partial^2 f/\partial s^2|_{s=0} = \alpha(1-\alpha).
\end{align}
Since $\alpha(1-\alpha)\leq 0.25$, we know the minimal universal $\sigma$ for all $\alpha$ is $0.5$, consistent with \eqref{eq:Phi_norm_0.5}.

For a specific $\alpha$, the minimal $\sigma$ that makes \eqref{eq:f} valid is denoted as $\sigma_{\Phi0}(\alpha)$, which is defined to be the sub-Gaussian norm of $\Phi(X_1^n)$ under the law $\Phi_\#P_0^{\otimes n}$, the push forward probability measure of $P_0^{\otimes n}$ induced by $\Phi$. We may also simply state that $\sigma_{\Phi0}(\alpha)$ is the sub-Gaussian norm of $\Phi(X_1^n)$ under $H_0$. The norm $\sigma_{\Phi0}(\alpha)$ can be numerically obtained in a principled way, as summarized in the following theorem.
\begin{thm}\label{thm:Phi_norm}
For $\alpha\neq 0.5$, besides the trivial solution $(\sigma,0)$ with any $\sigma>0$, the equations
\begin{align}\label{eq:norm_condition}
\left\{
  \begin{array}{ll}
    f = \frac{1}{2}\sigma^2s^2, \\
    \frac{\partial f}{\partial s} = \sigma^2s,
  \end{array}
\right.
\end{align}
have only one nontrivial solution $(\sigma^\ast,s^\ast)$ where $s^\ast\neq 0$. The sub-Gaussian norm of $\Phi(X_1^n)$ under $H_0$ is $\sigma_{\Phi0}=\sigma^\ast$. For $\alpha=0.5$, $\sigma_{\Phi0} = 0.5$.
\end{thm}

\begin{proof}
We will consider three cases based on the value of $\alpha$.

\emph{Case I: $\alpha=0.5$.} In this case, $\sigma_{\Phi0}$ can be obtained directly by noticing
\begin{align}
\exp[f] &= \cosh\left(\frac{s}{2}\right) = \sum_{n=0}^\infty\frac{(s/2)^{2n}}{(2n)!} \leq \sum_{n=0}^\infty\frac{(s/2)^{2n}}{2^nn!}\notag\\
&= \exp\left(\frac{1}{2}\times 0.5^2\times s^2\right).\notag
\end{align}
Hence by direct inspection, $\sigma_{\Phi0} = 0.5$.

\emph{Case II: $0<\alpha<0.5$.} Before diving into the proof, we briefly address the main idea first. Given $\alpha$, the function $h$ depends on both $s$ and $\sigma$. Requiring its maximum to be no greater than 0 at some $\sigma$ naturally leads to two conditions that $h(s,\sigma)=0$ and $\partial h(s,\sigma)/\partial s=0$, which are just \eqref{eq:norm_condition}. It is expected that $\sigma_{\Phi0}$ can be obtained from the corresponding nontrivial solutions, since it is the minimal $\sigma$ that satisfies \eqref{eq:f}. Fig. \ref{fig:PhiNorm} confirms this intuition, where $\alpha=0.05$ is assumed for illustration. By tuning $\sigma$ to some $\sigma^\ast$, one can see the maximum of $h$ at some $s^\ast>0$ can be exactly equal to 0, i.e., $h(s^\ast,\sigma^\ast)=0$. Also at this $s^\ast$, $h$ is tangent to the $s$-axis, indicating that $\partial h(s,\sigma^\ast)/\partial s|_{s=s^\ast}=0$. Hence $\sigma_{\Phi0}=\sigma^\ast$.

\begin{figure}[tb]
  \centerline{\includegraphics[width=1.0\linewidth]{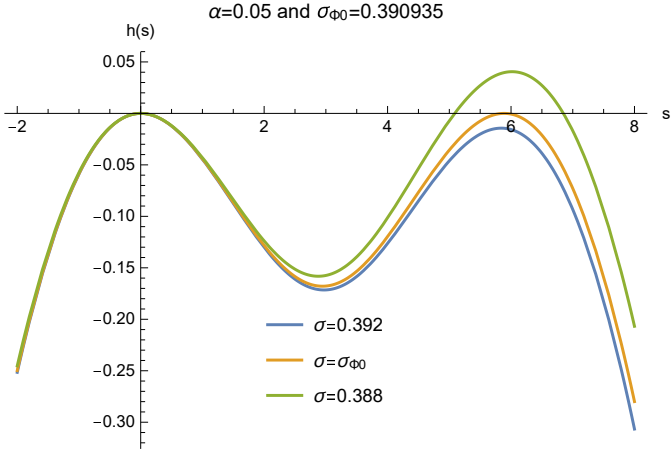}}
  \caption{Assuming $\alpha=0.05$, we show the main idea underlying Theorem \ref{thm:Phi_norm}, and numerically calculate $\sigma_{\Phi0}$, which is the minimal $\sigma$ such that $h(s)$ is no greater than 0 for all $s$ due to the sub-Gaussian property \eqref{eq:f}.} \label{fig:PhiNorm}
\end{figure}

Now we turn to the proof. It is trivial that for any $\alpha$, $h$ attains its maximal value 0 at $s=0$, no matter what $\sigma>0$ is. This does not provide much useful information of $\sigma_{\Phi0}$. To proceed, we need a nontrivial local maximum of $h(s)$ at some $s\neq 0$. Our first observation is that when $0<\alpha<0.5$, there is no local maximum achieved for $s<0$, because $\partial h/\partial s>0$ for all $s<0$. To see this, let $a \equiv e^{-|s|(1-\alpha)}$, $b\equiv e^{|s|\alpha}$, and $\delta\equiv0.5-\alpha$, then we have that
\begin{align}
\frac{\partial h}{\partial s} &= -\alpha(1-\alpha)\frac{b-a}{\alpha a + (1-\alpha)b} +\sigma^2|s|\notag\\
&= -\alpha(1-\alpha)\frac{1-e^{-|s|}}{(0.5+\delta)+(0.5-\delta)e^{-|s|}}+\sigma^2|s|\notag\\
&> -2\alpha(1-\alpha)\tanh(|s|/2) +\sigma^2|s|\notag\\
&> [\sigma^2-\alpha(1-\alpha)]\times|s|\geq 0,\notag
\end{align}
where $\alpha<0.5$ (hence $\delta>0$) is used in the first inequality, the second inequality is due to $\tanh(x)<x$ for $x>0$, and the last inequality is given by \eqref{eq:sigma_range} since we have already known $\Phi(X_1^n)$ is sub-Gaussian. This result indicates that the nontrivial maximum, if any, can only be found at some $s>0$.

For $s>0$, following similar steps, we obtain
\begin{align}
\frac{\partial h}{\partial s}=\frac{\alpha(1-\alpha)}{\frac{1}{2}\coth\left(\frac{s}{2}\right) - \left(\frac{1}{2}-\alpha\right)} - \sigma^2s,\notag
\end{align}
and the condition $\partial h/\partial s=0$ then implies
\begin{align}
g(s)\equiv\frac{s}{2}\coth\left(\frac{s}{2}\right) = \left(\frac{1}{2}-\alpha\right)s+\frac{\alpha(1-\alpha)}{\sigma^2}\equiv l(s,\sigma).\notag
\end{align}
It is straightforward to check that $g(s)\geq 1$ is a positive, monotonically increasing, and strongly convex function. Hence it can intersect the straight line $l(s,\sigma)$ at no more than two points. Note $g(0^+)=1$, and $g'(0^+)=0$. The intercept of $l(s,\sigma)$ is $\alpha(1-\alpha)/\sigma^2\in(0,1)$ by \eqref{eq:sigma_range}, and the slope is greater than 0. Hence by tuning $\sigma$, it is always possible to make $g(s)$ and $l(s,\sigma)$ intersect twice. Denote these two points as $s_1(\sigma)$ and $s_2(\sigma)$, respectively, with $h(s_1)<h(s_2)$. As shown in Fig. \ref{fig:PhiNorm}, $h(s_1)$ is the minimum between two maxima $h(s=0)$ and $h(s_2)$. Then further requiring $h(s_2(\sigma))=0$ at some $\sigma^\ast$, which is attainable since $\Phi$ is known to be sub-Gaussian, we obtain $\sigma_{\Phi0}=\sigma^\ast$, and Theorem \ref{thm:Phi_norm} for the $0<\alpha<0.5$ part is proved.

\emph{Case III: $0.5<\alpha<1$.} Note $f(s)$ or $h(s)$ is invariant under the transformations $\alpha \leftrightarrow 1-\alpha$ and $s \leftrightarrow -s$. Hence $\sigma_{\Phi0}$ is the same for $\alpha$ and $1-\alpha$.

Combining all three cases, we have proved Theorem \ref{thm:Phi_norm}.
\end{proof}

\begin{figure}[tb]
    \centerline{\includegraphics[width=1.0\linewidth]{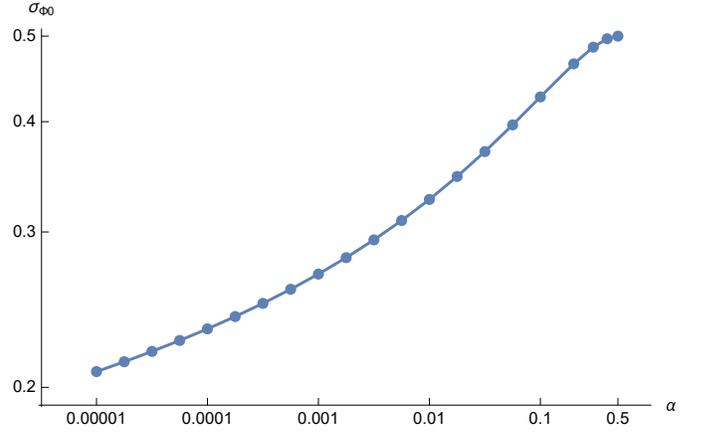}}
    \caption{The sub-Gaussian norm $\sigma_{\Phi0}$ is plotted as a function of type I error $\alpha$. Since $\sigma_{\Phi0}$ is the same for $\alpha$ and $1-\alpha$, we only plot the result for $\alpha\in(0,0.5]$.} \label{fig:alphasigma}
\end{figure}

One can calculate $\sigma_{\Phi0}$ in a principled way with given $\alpha$, without knowing $P_0$ or $P_1$ or the constant $c$ in the test function. We summarize the relation between $\alpha$ and $\sigma_{\Phi0}$ for \eqref{eq:Phi} in Fig. \ref{fig:alphasigma}. Error bounds for hypothesis testing can now be established based on $\sigma_{\Phi0}$.

\subsection{Sub-Gaussian bound for binary hypothesis testing}
\begin{lem}\label{lem:inequality}
Consider two general probability measures $\nu$ and $\mu$ on a common measurable space. Suppose $\nu\ll\mu$, and let $g\equiv d\nu/d\mu$ be the density of $\nu$ with respect to $\mu$. Let $Y$ be a sub-Gaussian random variable which is a function of $X$ that has law $\mu$ or $\nu$. Then we have
\begin{align}\label{eq:general_inequality}
|\mathbb{E}_{X\sim\nu} Y - \mathbb{E}_{X\sim\mu} Y|
\leq \sigma_{Y_{\#}\mu}\sqrt{2D_\mathrm{KL}(\nu\rVert\mu)},
\end{align}
where $\sigma_{Y_{\#}\mu}$ denotes the sub-Gaussian norm of $Y$ with respect to the push forward measure $Y_{\#}\mu$.
\end{lem}

Recently, there have been several works with findings similar to Theorem \ref{thm:bound}, in the context of nonequilibrium statistical physics \cite{me}, data exploration or model bias analysis \cite{IEEE:Zou,IEEE:Wang}, or uncertainty quantification for stochastic processes \cite{SIAM:Rey-Bellet}. They can, however, be analyzed in a unified way based on the spirit in \cite{JFA:Bobkov}.

\begin{proof}
We have assumed $\nu\ll\mu$ and $g\equiv d\nu/d\mu$. The associated entropy functional of $g$ with respect to $\mu$ is defined as $\mathrm{Ent}_{\mu}(g)=\int g\ln g d\mu$. It is straightforward to find that
\begin{align}\label{eq:KL_Ent}
\mathrm{Ent}_{\mu}(g)
&= \int \frac{d\nu}{d\mu}\ln\left(\frac{d\nu}{d\mu}\right) d\mu 
= \int \ln\left(\frac{d\nu}{d\mu}\right) d\nu\notag\\
&= D_\mathrm{KL}(\nu\rVert\mu).
\end{align}
On the other hand, by the variational representation of $\mathrm{Ent}_{\mu}(g)$, we have that
\begin{align}\label{eq:Ent_Variational}
\mathrm{Ent}_{\mu}(g) = \sup_\eta \int \eta gd\mu,\ \text{with}\ \int e^\eta d\mu\leq 1,
\end{align}
where $\eta$ is a measurable function. We have $\int \eta gd\mu = \mathbb{E}_\mu \eta g = \mathbb{E}_\nu \eta$ and $\int e^\eta d\mu = \mathbb{E}_\mu e^\eta$.

By assumption, $Y(X)$ is sub-Gaussian. If $X\sim\mu$, then the sub-Gaussian norm of $Y$ under the push forward measure $Y_\#\mu$ is $\sigma_{Y_{\#}\mu}$. Let us construct $\eta$ as
\begin{align}
\eta=s[Y(X) -\mathbb{E}_{X\sim\mu} Y(X)]-\frac{1}{2}\sigma_{Y_{\#}\mu}^2s^2.\notag
\end{align}
It is clear that $\mathbb{E}_\mu e^\eta\leq 1$ can be satisfied. Combining $\eta$ with \eqref{eq:KL_Ent} and \eqref{eq:Ent_Variational}, we arrive at
\begin{align}
D_\mathrm{KL}(\nu\rVert\mu)
&\geq  \mathbb{E}_{X\sim\mu} \eta(Y(X))g\notag\\
&=  \mathbb{E}_{X\sim\mu} g\left[s(Y-\mathbb{E}_{X\sim\mu} Y)-\frac{1}{2}\sigma_{Y_{\#}\mu}^2s^2\right]\notag\\
&= \mathbb{E}_{X\sim\nu} \left[s(Y-\mathbb{E}_{X\sim\mu} Y)-\frac{1}{2}\sigma_{Y_{\#}\mu}^2s^2\right]\notag\\
&= s(\mathbb{E}_{X\sim\nu} Y - \mathbb{E}_{X\sim\mu} Y) -\frac{1}{2}\sigma_{Y_{\#}\mu}^2s^2,\notag
\end{align}
which holds for any $s\in\mathbb{R}$. Hence, Lemma \ref{lem:inequality} is proved:
\begin{align}
|\mathbb{E}_{X\sim\nu} Y - \mathbb{E}_{X\sim\mu} Y|
&\leq \left [\inf_{|s|} \frac{D_\mathrm{KL}(\nu\rVert\mu)}{|s|} + \frac{1}{2}\sigma_{Y_{\#}\mu}^2|s|\right]\notag\\
&= \sigma_{Y_{\#}\mu}\sqrt{2D_\mathrm{KL}(\nu\rVert\mu)}.\notag
\end{align}
\end{proof}

\begin{thm}\label{thm:bound}
Suppose $P_1\ll P_0$, and denote the sub-Gaussian norm of test function \eqref{eq:Phi} under the null hypothesis $H_0$ as $\sigma_{\Phi0}$. Then we have
\begin{align}\label{eq:our_bound}
|\mathbb{E}_{X_1^n\sim P_0^{\otimes n}}\Phi(X_1^n) -
& \mathbb{E}_{X_1^n\sim P_1^{\otimes n}} \Phi(X_1^n)| \notag\\
&\leq \sigma_{\Phi0}\sqrt{2nD_\mathrm{KL}(P_1\rVert P_0)}.
\end{align}
\end{thm}
\begin{proof}
Let $X=X_1^n$, $\nu=P_1^{\otimes n}$ and $\mu=P_0^{\otimes n}$, and due to the i.i.d. setting, $D_\mathrm{KL}(P_1^{\otimes n}\rVert P_0^{\otimes n}) = nD_\mathrm{KL}(P_1\rVert P_0)$. Then the proof is completed by letting $Y=\Phi(X_1^n)$ in Lemma \ref{lem:inequality}.
\end{proof}

\begin{cor}\label{cor:error_sum}
One has
\begin{align}\label{eq:error_sum}
\alpha + \beta \geq
1 - \sigma_{\Phi0}\sqrt{2nD_\mathrm{KL}(P_1\rVert P_0)}.
\end{align}
\end{cor}
\begin{proof}
Insert definitions \eqref{eq:alpha} and \eqref{eq:beta} into \eqref{eq:our_bound} and then simplify to obtain the result.
\end{proof}

\begin{remark}
Corollary \ref{cor:error_sum} can be relaxed by replacing the sub-Gaussian norm $\sigma_{\Phi0}$ with one of its upper bounds. In fact, if we use the universal upper bound provided by \eqref{eq:Phi_norm_0.5}, then Corollary \ref{cor:error_sum} reduces to Pinsker's classical result \eqref{eq:Pinsker_result}. However, our bound is always stronger in general. In particular, when controlling $\alpha$ is more important than controlling $\beta$, one might set $c>0$ to put more emphasis on it. Hence for the same sample size $n$, the larger $c$ is, the smaller $\alpha$ and $\sigma_{\Phi0}$ are, resulting in a tighter bound for $\beta$.
\end{remark}

\begin{remark}
There is another inequality from Theorem \ref{thm:bound} that $\alpha +\beta\leq 1+\sigma_{\Phi0}\sqrt{2nD_\mathrm{KL}(P_1\rVert P_0)}$. But it is somewhat trivial because the bound is greater than 1 and in general does not provide much useful information. For example, one can always accept $H_0$, and for this trivial decision rule, $\alpha=0$, but $\beta\leq 1$ by definition. Hence $\alpha+\beta\leq 1$, and the extra term $\sigma_{\Phi0}\sqrt{2nD_\mathrm{KL}(P_1\rVert P_0)}$ is not informative at all.  
\end{remark}

\begin{remark}
Suppose also $P_0\ll P_1$, which is the usual case in hypothesis testing. Then by symmetry, it is straightforward to have
\begin{align}\label{eq:alpha_bound}
\alpha + \beta \geq
1 - \sigma_{\Phi1}\sqrt{2nD_\mathrm{KL}(P_0\rVert P_1)},
\end{align}
where $\sigma_{\Phi1}$ is the sub-Gaussian norm of $\Phi(X_1^n)$ under $H_1$, and it is a function of $\beta$. This result is nontrivially different than \eqref{eq:error_sum}, not only because different norms are applied, but also because the KL divergence is not symmetric in two involved probabilities. Given \eqref{eq:alpha_bound}, we can either bound $\alpha$ when $\beta$ is given or bound $\beta$ in an implicit way when $\alpha$ is given.
\end{remark}

\begin{remark}
Similar to \eqref{eq:Pinsker_result}, our bound is also nonasymptotic in nature as it holds for any finite $n$. The expense we pay for this, however, is that in the large $n$ and small $\alpha$ limit, our bound for $\beta$ is not as tight as Stein's lemma which states that $\beta\sim e^{-nD_\mathrm{KL}(P_0\rVert P_1)}$ \cite{Book:Cover}.
\end{remark}

\subsection{Sub-Gaussian bound for $M$-ary hypothesis testing}

A generalization of our result to the $M$-ary hypothesis testing can be obtained. Suppose there are $M$ hypotheses, represented by the corresponding probability distributions $\{P_1,\ldots,P_M\}$. Suppose from one of such distributions $P_{i_0}$, $n$ data points $X_1^n$ are drawn independently. Our task is to infer the hypothesis index ${i_0}$ from data. Similar to \eqref{eq:Phi}, let us consider the test function for the $i$-th hypothesis as
\begin{align}\label{eq:varphi}
\varphi_i(X_1^n) = \prod_{j\neq i}I\{D_\mathrm{KL}(\hat P_n\rVert P_j)-D_\mathrm{KL}(\hat P_n\rVert P_i)>c_i\},\notag
\end{align}
with $\varphi = 1-\Phi$ in the binary case. We will consider the case that $c_i=0$ for all $i\in\{1,\ldots,M\}$. Unlike in the binary case where $c>0$ can be adopted to intentionally render a small $\alpha$, the test function here is purely likelihood-based without any prescribed preference over any particular hypothesis. It is known that this approach minimizes $\alpha+\beta$ in the binary case (the Bayes classifier). From $M$ such test functions $\varphi_i$, one can construct a random vector ${\boldsymbol \varphi}=(\varphi_1,\ldots,\varphi_M)$. Assume there always exists a single index $i_0$ such that
\begin{align}
D_\mathrm{KL}(\hat P_n\rVert P_{j})-D_\mathrm{KL}(\hat P_n\rVert P_{i_0}) > 0 \notag
\end{align}
holds for all $j\neq i_0$. In this case, $\varphi_{i_0}=1$, and $\varphi_{j}=0$ for $j\neq i_0$. Since $X_1^n$ is random, we expect that $i_0$ may differ for each realization. However, it is almost surely with respect to all $P_i$'s that
\begin{equation}
\sum_{i=1}^M \varphi_i(X_1^n) = 1. \label{eq:M_ary_sum}
\end{equation}

Under $M$ hypotheses, we can construct a matrix, denoted $\mathbb{E}\boldsymbol\varphi$, that encodes the error incurred in testing:
\begin{align}\label{eq:M_ary_Matrix}
\boldsymbol{\mathbb{E}\varphi} \equiv \left(
                                      \begin{array}{ccc}
                                        \mathbb{E}_1\varphi_1 & \cdots &  \mathbb{E}_1\varphi_M \\
                                        \vdots & \ddots & \vdots \\
                                        \mathbb{E}_M\varphi_1 & \cdots & \mathbb{E}_M\varphi_M \\
                                      \end{array}
                                    \right),
\end{align}
where the matrix element $\mathbb{E}_i\varphi_j\equiv\mathbb{E}_{X_1^n\sim P_i^{\otimes n}}\varphi_j$. By \eqref{eq:M_ary_sum}, the row sum of $\mathbb{E}\boldsymbol\varphi$ is 1. The diagonal elements of $\mathbb{E}\boldsymbol\varphi$ are actually the probabilities that the underlying hypothesis is correctly identified. In other words, the probability of making an incorrect decision when the data are generated from the $i$th hypothesis is $\alpha_i \equiv 1-\mathbb{E}_i\varphi_i$. We denote $\alpha_{\max}\equiv\max_{i}\alpha_i$. The following theorem provides a lower bound to $\alpha_{\max}$ that is complementary to Fano's inequality.

\begin{thm}\label{thm:beyond_fano}
Suppose $P_i\ll P_j$ for all $i,j\in\{1,\ldots, M\}$. For any $j\in\{1,\ldots,M\}$, we have
\begin{align}\label{eq:beyond_fano}
\alpha_{\max} \geq 1 - \frac{1}{M} - \frac{1}{M}\sum_{i=1}^M \sigma_{\varphi_i}\sqrt{2nD_\mathrm{KL}(P_j\rVert P_i)},
\end{align}
where $\sigma_{\varphi_i}$ is the sub-Gaussian norm of $\varphi_i$ with respect to the $i$th hypothesis.
\end{thm}

\begin{proof}
First note that $\varphi_i$ is sub-Gaussian since it takes on values in $\{0,1\}$. If $\alpha_i$ is fixed, then the sub-Gaussian norm $\sigma_{\varphi_i}$ can be calculated similarly as in the binary case. Even $\alpha_i$ is unknown, by \eqref{eq:our_bound}, we can formally have 
\begin{align}\label{eq:Ei_varphi}
\mathbb{E}_i\varphi_i \leq \sigma_{\varphi_i}\sqrt{2nD_\mathrm{KL}(P_j\rVert P_i)} + \mathbb{E}_j\varphi_i.
\end{align}
Summing over $i$ and combining \eqref{eq:M_ary_sum}, we find
\begin{align}
M(1-\alpha_{\max}) \leq \sum_{i=1}^M \mathbb{E}_i\varphi_i \leq 1 + \sum_{i=1}^M \sigma_{\varphi_i}\sqrt{2nD_\mathrm{KL}(P_j\rVert P_i)}.\notag
\end{align}
Finally, we arrive at \eqref{eq:beyond_fano} by rearranging the terms. Hence the proof is completed.
\end{proof}

If we aim at lower bounding $\alpha_{\max}$, then using the sub-Gaussian norm $\sigma_{\varphi_i}$ in Theorem \ref{thm:beyond_fano} seems not useful practically, since $\sigma_{\varphi_i}$ itself depends on $\alpha_i$. Nonetheless, due to the universal upper bound \eqref{eq:Phi_norm_0.5}, we can have a relaxed version of \eqref{eq:beyond_fano} as in the corollary below.
\begin{cor}
For any $j\in\{1,\ldots, M\}$, we have
\begin{align}
\alpha_{\max} \geq 1 - \frac{1}{M} - \frac{1}{M}\sum_{i=1}^M \sqrt{\frac{n}{2}D_\mathrm{KL}(P_j\rVert P_i)},
\end{align}
or, using the mean square root of KL divergences, we have
\begin{align}
\alpha_{\max} \geq 1 - \frac{1}{M} - \frac{1}{M^2}\sum_{i,j=1}^M \sqrt{\frac{n}{2}D_\mathrm{KL}(P_j\rVert P_i)}.
\end{align}
Furthermore, if $D_\mathrm{KL}(P_i\rVert P_j)\leq \delta$ holds for each pair of $i$ and $j$, then
\begin{align}\label{eq:alpha_max}
\alpha_{\max} \geq 1 - \frac{1}{M} - \sqrt{\frac{n}{2}\delta}.
\end{align}
\end{cor}

\begin{remark}
It is interesting to compare \eqref{eq:alpha_max} with Fano's inequality \cite{Book:Fano}, which, under the same assumption that all KL divergences are uniformly bounded by $\delta$, states that
\begin{align}\label{eq:Fano}
\alpha_{\max}^{\mathrm{Fano}} \geq 1 - \frac{n\delta+\ln 2}{\ln(M-1)}.
\end{align}
As evidenced by the scalings of $M$ and $n$ in \eqref{eq:alpha_max} and \eqref{eq:Fano}, respectively, there is a region that our result outperforms Fano's in the sense that it provides a greater lower bound for $\alpha_{\max}$. Qualitatively, this happens when at least one of the number of hypotheses $M$ and the sample size $n$ is not large. For example, when $M=3$, Fano's inequality is trivial, while our result can still work nontrivially. 

\end{remark}

\section{Conclusion and discussion}
In this work, by using the sub-Gaussian property of test functions, we uncover two universal error bounds in terms of the sub-Gaussian norm and the Kullback-Leibler divergence. In the case of binary hypothesis testing, our bound \eqref{eq:error_sum} is always tighter than Pinsker's bound \eqref{eq:Pinsker_result} for any given $\alpha\neq 0.5$. In the case of $M$-ary hypothesis testing, our result \eqref{eq:beyond_fano} is complementary to Fano's inequality \eqref{eq:Fano} by providing a more informative bound when the number of hypotheses or the sample size is not large.

Given the universality of our results, we hope, with possible generalizations, they can find potential applications in fields ranging from clinical trials to quantum state discrimination. In particular, the quantum extension of these bounds is of special interest. Due to the experimental cost, it may be important to quantify statistical errors in the presence of a \textit{limited} number of observations, and nonasymptotic rather than asymptotic results are thus more relevant. Both our bounds hold for any finite sample size, and can hopefully be helpful in such cases.

\section*{Acknowledgment}
YW gratefully thank Prof. Dan Nettleton for helpful discussions that stimulated this work and Prof. Huaiqing Wu for a careful review of the manuscript and insightful comments.

\end{document}